\documentclass[12pt]{article}

  \usepackage{amsthm,latexsym,amssymb,amsfonts,amsmath,mathrsfs,float}
  \usepackage{xcolor}
  \usepackage{tikz, graphics, enumerate}
  \usepackage{hyperref}
  \usepackage{fullpage}

  \renewcommand{\Pr}{\mbox{\rm Pr}}
  \newcommand{\Exp}{{\mathbb{E}}}
  \newcommand{\E}{\mathbb{E}}

  \newcommand{\R}{\mathbb{R}} 
  \newcommand{\N}{\mathbb{N}} 
  \newcommand{\Z}{\mathbb{Z}} 
  \newcommand{\F}{\mathbb{F}} 
  \newcommand{\pmset}[1]{\{-1,1\}^{#1}} 
  \newcommand{\bset}[1]{\{0,1\}^{#1}} 

  \DeclareMathOperator{\Anorm}{sp}
  \newcommand{\st}{:\,} 

  \newcommand{\eps}{\varepsilon}

  \newcommand{\ceil}[1]{\lceil{#1}\rceil}

  \DeclareMathOperator{\spec}{Spec}

\newcommand{\inpro}[2]{\langle #1,#2 \rangle}
\newcommand{\norm}[1]{\left\lVert #1\right\rVert}

\newcommand{\mathbbm}[1]{\mathbb{1}}
\newcommand{\specnorm}[1]{\left\lVert #1 \right\rVert_{\mathrm{sp}}}
\def\F{{\mathbb{F}}}

\def\Z{{\mathbb{Z}}}
\def\N{{\mathbb{N}}}
\def\R{{\mathbb{R}}}

\def\E{{\mathbb E}}

\def\cA{{\mathcal A}}

\def\sBC{{\{-1,1\}}}

\def\vc{{\mathrm{vc}}}
\def\PF{{\mathbb{P}\mathbb{F}}}

  \newcommand{\beq}{\begin{equation}}
  \newcommand{\eeq}{\end{equation}}
  \newcommand{\beqn}{\begin{equation*}}
  \newcommand{\eeqn}{\end{equation*}}
  \newcommand{\beqr}{\begin{eqnarray}}
  \newcommand{\eeqr}{\end{eqnarray}}
  \newcommand{\beqrn}{\begin{eqnarray*}}
  \newcommand{\eeqrn}{\end{eqnarray*}}
  \newcommand{\bmline}{\begin{multline}}
  \newcommand{\emline}{\end{multline}}
  \newcommand{\bmlinen}{\begin{multline*}}
  \newcommand{\emlinen}{\end{multline*}}

  \theoremstyle{plain}
  \newtheorem{theorem}{Theorem}[section]
  \newtheorem{lemma}[theorem]{Lemma}

  \newtheorem{corollary}[theorem]{Corollary}
  \newtheorem{question}[theorem]{Question}
  \theoremstyle{definition}
  \newtheorem{definition}[theorem]{Definition}

  \theoremstyle{remark}
  
  \renewenvironment{proof}[1][]{
    	\begin{trivlist}
     	\item[\hspace{\labelsep}{\em\noindent Proof#1:\/}]}
     	{{\hfill$\Box$}
    	\end{trivlist}
  }
  
  \makeatletter
  \newtheorem*{rep@theorem}{\rep@title}
  \newcommand{\newreptheorem}[2]{%
  \newenvironment{rep#1}[1]{%
  \def\rep@title{#2 \ref{##1}}%
  \begin{rep@theorem}}%
  {\end{rep@theorem}}}
  \makeatother

  \newreptheorem{lemma}{Lemma}
  \newreptheorem{theorem}{Theorem}
  \newreptheorem{corollary}{Corollary}
  \newreptheorem{proposition}{Proposition}
  \newreptheorem{conjecture}{Conjecture}

\begin{document}

\title{Outlaw distributions and locally decodable codes\thanks{A preliminary version of this paper appeared in the proceedings of ITCS 2017~\cite{BDG_ITCS:2017}}}
\author{
Jop Bri\"{e}t\thanks{
CWI, Science Park 123, 1098 XG Amsterdam, the Netherlands.
Supported by a VENI grant and the Gravitation-grant NETWORKS-024.002.003 from the Netherlands Organisation for Scientific Research~(NWO). Email: \texttt{j.briet@cwi.nl}.} 
\and 
Zeev Dvir\thanks{Department of Computer Science and Department of Mathematics, Princeton University, Princeton, NJ 08540, USA.  Research supported by NSF grant CCF-1523816 and by the Sloan foundation. Email: \texttt{zeev.dvir@gmail.com}.} 
\and 
Sivakanth Gopi\thanks{Department of Computer Science, Princeton University, Princeton, NJ 08540, USA.  Research supported by NSF grant CCF-1523816 and by the Sloan foundation. Email: \texttt{sgopi@cs.princeton.edu}.} 
}
\date{}
\maketitle
\begin{abstract}
Locally decodable codes (LDCs) are error correcting codes that allow for decoding of a single message bit using a small number of queries to a corrupted encoding. Despite decades of study, the optimal trade-off between query complexity and codeword length is far from understood.  In this work, we give a new characterization of LDCs  using distributions over Boolean functions whose expectation is hard to approximate (in~$L_\infty$~norm) with a small number of samples. We coin the term `outlaw distributions' for such distributions since they `defy' the Law of Large Numbers. We show that the existence of outlaw distributions over sufficiently `smooth' functions implies the existence of constant query LDCs and vice versa. We give several candidates for outlaw distributions over smooth functions coming from finite field incidence geometry, additive combinatorics and from hypergraph (non)expanders.

We also prove a useful lemma showing that (smooth) LDCs which are only required to work on average over a random message and a random message index can be turned into true LDCs at the cost of only constant factors in the parameters.
\end{abstract}

\section{Introduction}

Error correcting codes (ECCs) solve the basic problem of communication over noisy channels. 
They encode a message into a codeword from which, even if the channel partially corrupts it, the message can later be retrieved.
With one of the earliest applications of the \emph{probabilistic method}, formally introduced by  Erd\H{o}s in 1947, pioneering work of Shannon~\cite{Shannon:1948} showed the existence of optimal (capacity-achieving) ECCs.
The problem of explicitly constructing such codes has fueled the development of coding theory ever since.
Similarly, the exploration of many other fascinating structures, such as Ramsey graphs, expander graphs, two source extractors, etc., began with a striking existence proof  via the probabilistic method, only to be followed by decades of catch-up work on explicit constructions.
Locally decodable codes (LDCs) are a special class of error correcting codes whose development has not followed this line.
The defining feature of LDCs is that they allow for ultra fast decoding of single message bits, a property that typical ECCs lack, as their decoders must read an entire (possibly corrupted) codeword to achieve the same.
They were first formally defined in the context of channel coding in~\cite{Katz:2000}, although they (and the closely related locally correctable codes) implicitly appeared in several previous works in other settings, such as program checking~\cite{BK95}, probabilistically checkable proofs~\cite{AS98,ALMSS98} and private information retrieval schemes (PIRs)~\cite{CKGS98}. More recently, LDCs have even found applications in Banach-space geometry~\cite{Briet:2012d} and LDC-inspired objects called local reconstruction codes found applications in fault tolerant distributed storage systems~\cite{GHSY12}. See \cite{Yekhanin:2012} for a survey of LDCs and some of the applications.

Despite their many applications, our knowledge of LDCs is very limited; the best-known constructions are far from what is currently known about their limits.
Although standard random (linear) ECCs do allow for some weak local-decodability, they are outperformed by even the earliest explicit constructions~\cite{KS07}.
All the known constructions of LDCs 
were obtained by explicitly designing such codes using some algebraic objects like low-degree polynomials or matching vectors~\cite{Yekhanin:2012}. 

In this paper, we give a characterization of LDCs in probabilistic and geometric terms, making them amenable to probabilistic constructions. On the flip side, these characterizations might also be easier to work with for the purpose of showing lower bounds. We will make this precise in the next section.
Let us first give the formal definition of an LDC.

\begin{definition}[Locally decodable code]\label{ldc:def}
For positive integers~$k, n,q$ and parameters $\eta,\delta\in (0,1/2]$, a map $C: \bset{k}\to\bset{n}$ is a $(q,\delta,\eta)$-\emph{locally decodable code} if, for every~$i\in[k]$, there exists a randomized \emph{decoder} (a probabilistic algorithm)~$\mathcal{A}_i$  such that:
 
\begin{itemize}
\item For every message $x\in\bset{k}$ and string~$y\in\bset{n}$ that differs from the codeword~$C(x)$ in at most~$\delta n$ coordinates,
\beq\label{eq:ldcdec}
\Pr[\mathcal A_i( y) = x_i] \geq \frac{1}{2} + \eta.
\eeq
\item The decoder $\mathcal{A}_i$ (non-adaptively) queries at most $q$ coordinates of $y$.\footnote{
We can assume that on input~$y\in\bset{n}$, the decoder~$\mathcal A_i$ first samples a set~$S\subseteq[n]$ of at most~$q$ coordinates according to a probability distribution depending on~$i$ only and then returns a random bit depending only on~$i$,~$S$ and the values of~$y$ at~$S$.
}
\end{itemize}
\end{definition}

\paragraph{Known results}
The main parameters of LDCs are the number of queries $q$ and the length of the encoding~$n$ as a function of $k$ and $q$, typically the parameters $\delta,\eta$ are fixed constants. The simplest example is the Hadamard code, which is a 2-query LDC with $n=2^k$. The 2-query regime is the only nontrivial case where optimal lower bounds are known: it was shown in ~\cite{Kerenidis:2004,GKST02} that exponential length is necessary. In general, Reed-Muller codes of degree $q-1$ are $q$-query LDCs of length $n=\exp(O(k^{1/q-1}))$. For a long time, these were the best constructions for constant $q$, until in a breakthrough work by~\cite{Yekhanin:2007,Efremenko:2009}, $3$-query LDCs were constructed with subexponential length $n=\exp(\exp(O(\sqrt{\log k})))$. More generally they constructed $2^r$-query LDCs with length $n=\exp(\exp(O(\log^{1/r} k)))$. For $q\geq 3$, the best-known lower bounds leave huge gaps, giving only polynomial bounds. A $3$-query LDC must have length $n\ge \tilde{\Omega}(k^2)$, and more generally for a $q$-query LDC, $n\ge \tilde{\Omega}(k^{1+1/(\ceil{q/2}-1)})$~\cite{Kerenidis:2004,Woodruff:2007a}\footnote{We use $\tilde{\Omega}(\ ), \tilde{O}(\ ), \tilde{\omega}(\ ), \tilde{o}(\ )$ to hide polylogarithmic factors through out this paper.}. LDCs where the codewords are over a large alphabet are studied because of their relation to private information retrieval schemes~\cite{CKGS98,Katz:2000}. In~\cite{DG15}, 2-query LDCs of length $n=\exp(k^{o(1)})$ over an alphabet of size $\exp(k^{o(1)})$ were constructed.
There is also some exciting recent work on LDCs where the number of queries grows with~$k$, in which case there are explicit constructions with constant-rate (that is, $n = O(k)$)  and query complexity $q=\exp(O(\sqrt{\log n}))$; in fact we can even achieve the optimal rate-distance tradeoff of traditional error correcting codes~\cite{KSY10,KMRS16,GKORS16}. We cannot yet rule out the exciting possibility that constant rate LDCs with polylogarithmic query complexity exist.

%
%

%
\subsection{LDCs from distributions over smooth Boolean functions}

Our main result shows that LDCs can be obtained from ``outlaw'' distributions over ``smooth'' functions.
The term outlaw refers to the Law of Large Numbers, which says that the average of independent samples tends to the expectation of the distribution from which they are drawn.
Roughly speaking, a probability distribution is an outlaw if many samples are needed for  a good estimation of the expectation
and
a smooth function over the $n$-dimensional Boolean hypercube is one that has no influential variables.
Paradoxically, while many instances of the probabilistic method use the fact that sample means of a small number of independent random variables tend to concentrate around the true mean, as captured for example by the Chernoff bound, our main result requires precisely the opposite. 
We show that if \emph{at least} $k$ samples from a distribution over smooth functions are needed to approximate the mean, then there exists an $O(1)$-query LDC sending $\bset{\Omega(k)}$ to $\bset{n}$, where the hidden constants depend only the smoothness and mean-estimation parameters.

To make this precise, we now formally define smooth functions and outlaw distributions. 
Given a function $f:\pmset{n}\to \R$, its \emph{spectral norm} 
is defined as
$$
\specnorm{f} =\sum_{S\subset [n]}|\widehat{f}(S)|,
$$
where $\widehat{f}(S)$ are the Fourier coefficients of $f$ (see Section~\ref{sec:prelims} for basics on Fourier analysis). 
We also consider the supremum norm, $\norm{f}_{L_\infty}=\sup\{|f(x)| \st x\in \pmset{n}\}.$
It follows from the Fourier inversion formula that $\norm{f}_{L_\infty}\le \specnorm{f}$. The \emph{$i$th discrete derivative} of $f$ is the function $(D_if)(x) = (f(x) - f(x^i))/2$, where $x^i$ is the point that differs from~$x$ on the $i$th coordinate. 
Smooth functions are functions whose discrete derivatives have small spectral norms.

\begin{definition}[$\sigma$-smooth functions]
For~$\sigma>0$, a function $f:\sBC^n\to \R$ is $\sigma$-smooth, if for every $i\in [n]$, we have $\specnorm{D_if}\le \sigma/n$.
\end{definition}

Intuition for the above definition may be gained from the fact that smooth functions have no influential variables.
The influences, $(\Exp_{x\in\pmset{n}}[(D_if)(x)^2])^{1/2}$, measure the extent to which changing the $i$th coordinate of a randomly chosen point changes the value of~$f$. Since $\norm{D_if}_{L_\infty}\le \specnorm{D_if}$, the directional derivatives of $\sigma$-smooth functions are uniformly bounded by $\sigma/n$, which is a much stronger condition than saying that the derivatives are small on average.
Outlaws are defined as follows.

\begin{definition}[Outlaw]\label{def:outlaw}
Let~$n$ be a positive integer and~$\mu$ be a probability distribution over real-valued functions on~$\pmset{n}$.
For a positive integer~$k$ and ${\eps > 0}$, say that~$\mu$ is a $(k, \eps)$-\emph{outlaw} if
for independent random $\mu$-distributed functions $f_1,\dots,f_{k}$ and $\bar f = \Exp_\mu[f]$,
\beqn
\Exp\left[\Big\|\frac{1}{k}\sum_{i=1}^k (f_i - \bar f)\Big\|_{L_\infty}\right] \geq \eps.
\eeqn
Denote by $\kappa_\mu(\epsilon)$ the largest integer~$k$ such that~$\mu$ is a $(k,\eps)$-outlaw.
\end{definition}

To approximate the true mean of an outlaw~$\mu$ to within~$\eps$ on average in the~$L_\infty$-distance, one thus needs~$\kappa_\mu(\eps)+1$ samples. Note that if $\mu$ is a distribution over $\sigma$-smooth functions, then the distribution $\tilde{\mu}$ obtained by scaling functions in the support of $\mu$ by $1/\sigma$ is a distribution over $1$-smooth functions and $\kappa_{\tilde{\mu}}(\eps/\sigma)=\kappa_{\mu}(\eps)$.

Our main result is then as follows.

\begin{theorem}[Main theorem]\label{thm:main}
Let $n$ be a positive integer and $\eps >0 $.
Let $\mu$ be a probability distribution over $1$-smooth functions on $\pmset{n}$ and $k = \kappa_\mu(\eps)$.
Then, there exists a $(q,\delta,\eta)$-LDC sending $\bset{l}$ to $\bset{n}$ where $l=\Omega(\eps^2 k/\log(1/\eps))$, $q=O(1/\eps)$, $\delta =\Omega(\eps)$ and $\eta=\Omega(\eps)$.
Additionally, if~$\mu$ is supported by degree-$d$ functions, then we can take $q=d$.
\end{theorem}

Note that the smoothness requirement is essential. For example the uniform distribution over the $n$ dictator functions $f_i(x)=x_i$ for $i\in [n]$ is an $(n/2,1)$-outlaw, but it cannot imply constant rate, constant query LDCs which we know do not exist. In fact we establish a converse to Theorem~\ref{thm:main}, showing that its hypothesis is essentially equivalent to the existence of LDCs in the small query complexity regime.

\begin{theorem}\label{thm:main_converse}
If $C:\bset{k}\to \bset{n}$ is a $(q,\delta, \eta)$-LDC, then there exists a probability distribution $\mu$ over $1$-smooth degree-$q$ functions on $\sBC^n$ such that $$\kappa_\mu(\eps)\ge \eta k$$ where $\eps=\eta\delta/(q2^{q/2})$.
\end{theorem}

Theorem~\ref{thm:main_converse} can in turn convert the problem of proving lower bounds on the length of LDCs to a problem on Banach space geometry.
In particular, for a distribution~$\mu$ over 1-smooth degree-$q$ functions on~$\bset{n}$, one can upper bound~$\kappa_\mu(\eps)$ in terms of type constants of the space of $q$-linear forms on~$\ell_q^{n+1}$~\cite{Briet:2015}.

\medskip
\paragraph{Candidate outlaws}
One scenario in which outlaw distributions can be obtained is using incidence geometry in finite fields. In particular, the  following result can be derived from our main theorem (stated a bit informally here, see Section~\ref{sec:lines} for the formal version).

\begin{corollary}\label{cor:lines_ldcs_informal}
Let $p>2$ be a fixed prime. Suppose that for a random set of directions $D\subset \F_p^{n}$  of size $|D|\le k$, with probability at least $1/2$, there exists a set $B\subset \F_p^n$  of size $|B|\ge \Omega(p^n)$ which does not contain any lines with direction in $D$.
Then, there exists a $p$-query LDC sending $\bset{\Omega(k)}$ to $\bset{p^n}$. 
\end{corollary}

The assumption in Corollary~\ref{cor:lines_ldcs_informal} that~$D$ be random is essential for it to be potentially interesting for LDCs.
If we instead ask that every set of directions~$D$ satisfies the condition---as we did in the conference version of this paper---then letting~$D$ be a subspace shows that~$k$ must be smaller than a constant depending only on~$p$ and~$\eps$ by Szemer\'edi's Theorem (Theorem~\ref{thm:szemeredi} below)~\cite{Fox:personal}.

The analogue of Corollary~\ref{cor:lines_ldcs_informal} in $\Z/N\Z$ where lines correspond to arithmetic progressions and directions correspond to common differences can also be used to construct LDCs. This question was studied in~\cite{FLW16}, where they show that if $D$ is a random subset of $\Z/N\Z$ of size $\omega(N^{1-1/q})$, then almost surely every dense subset of $\Z/N\Z$ contains a $q$-term arithmetic progression with common difference in $D$. 
Our main result, together with the best-known lower bounds on LDCs~\cite{Kerenidis:2004,Woodruff:2007a} can be used to improve the bounds of~\cite{FLW16} on random differences in Szemer\'edi's theorem as shown in the following corollary.
\begin{corollary}
For every $\alpha>0$, there is $N(\alpha)$ such that for every $N>N(\alpha)$ the following is true:
if $D\subset \Z/N\Z$ is a uniformly random subset of size $\tilde\omega(N^{1-1/\ceil{q/2}})$ when $q>2$ and $\omega(\log N)$ when $q=2$, then with probability $1-o(1)$, every subset $A\subset \Z/N\Z$ of size at least $\alpha N$ contains a $q$-term arithmetic progression with common difference in the set $D$.
\end{corollary}


Another setting in which our approach leads to interesting open problems is in relation to pseudorandom hypergraphs. Consider a partition of the complete bipartite graph $K_{n,n}$ into~$n$ perfect matchings. It is known that picking $k = O(\log n)$ of these matchings at random will give us a pseudorandom (expander) graph (of degree $k$). For some particular partitions (e.g., given by an Abelian group) 
 this bound is tight. The questions arising from our approach can be briefly summarized as follows: Can one find an $n$-vertex hypergraph~$H$ (say three uniform to be precise) and a partition of~$H$ into matchings so that, to get a pseudorandom hypergraph (defined appropriately) one needs  at least $k$  random matchings. This would give a code sending $\Omega(k)$-bit messages with encoding length $O(n)$ and so, becomes interesting when~$k$ is super poly-logarithmic in~$n$.
 We elaborate on this in Section~\ref{sec:ARLDCs}

\subsection{Techniques}

Our proof of Theorem~\ref{thm:main} proceeds in two steps.
The first step consists of turning an outlaw over smooth functions into a seemingly crude type of LDC that is only required to work on average over a uniformly distributed message and a uniformly distributed message index.
We call such codes \emph{average-case smooth codes} (see below).
The second step consists of showing that such codes are in fact not much weaker than honest LDCs. 

\medskip
\paragraph{From outlaws to average-case smooth codes}

The key ingredient for the first step is \emph{symmetrization}, a basic technique from high-dimensional probability.
We briefly sketch how this is used (we refer to Section~\ref{sec:outlaw2smoothcode} for the full proof).
Suppose that~$f_1,\dots,f_k$ are independent smooth functions distributed according to a $(k,\eps)$-outlaw with expectation~$\bar f$.
We introduce an independent copy\footnote{in this context sometimes referred to as a ``ghost copy'' as it will later disappear again} $f_i'$  of~$f_i$ for each $i\in[k]$ and consider the symmetrically distributed random functions~$f_i - f_i'$.
Since $\bar f = \Exp[f'_i]$ for each $i\in[k]$, Jensen's inequality and Definition~\ref{def:outlaw} imply that
\beqn
\Exp\big[\|(f_1 - f_1') + \cdots + (f_k - f_k')\|_{L_\infty}\big]
\geq
\Exp\big[\|(f_1 - \Exp[f_1']) + \cdots + (f_k - \Exp[f_k')\|_{L_\infty}\big]
\geq \eps k.
\eeqn
Since the random functions $f_i  -f_i'$ are independent and symmetric, we get that for independent uniformly random signs $x_1,\dots, x_k\in\pmset{}$, the above left-hand side equals
\beqn
\Exp\big[\|x_1(f_1 - f_1') + \cdots + x_k(f_k - f_k')\|_{L_\infty}\big].
\eeqn
The triangle inequality and the Averaging Principle
then give that there exist \emph{fixed} smooth functions $f_1^\star,\dots,f_k^\star$ such that on average over the random signs, we have
\beq\label{eq:eqeq}
\Exp\big[ \|x_1f_1^\star + \cdots + x_kf^\star_k\|_{L_\infty}\big] \geq \eps k/2.
\eeq
To get an average-case smooth code out of this, we view each sequence $x = (x_1,\dots,x_k)$ as a $k$-bit message and choose an arbitrary $n$-bit string for which the $L_\infty$-norm in~\eqref{eq:eqeq} is achieved to be the its encoding,~$C(x)$.
This gives a map $C:\pmset{k}\to\bset{n}$ satisfying
\beqn
\Exp\big[x_1f^\star_1(C(x)) + \cdots x_kf^\star_k(C(x))\big] \geq \eps k/2.
\eeqn
Equivalently, for uniform~$x$ and~$i$, we have $\Pr[f^\star_i(C(x)) = x_i] \geq \tfrac12+\tfrac\eps4$.
Finally, we use the smoothness property to transform the~$f_i^\star$ into decoders with the desired properties.
This is done in Section~\ref{sec:outlaw2smoothcode}.
It is in the application of the Averaging Principle where the probabilistic method appears in our construction of LDCs.
\medskip

\paragraph{Average-case smooth codes are LDCs}

Katz and Trevisan~\cite{Katz:2000} observed that LDC decoders must have the property that they select their queries according to distributions that do not favor any particular coordinate.
The intuition for this is that if they did favor a certain coordinate, then corrupting that coordinate would cause the decoder to err with too high a probability.
If instead, queries are sampled according to a ``smooth'' distribution, they will all fall on uncorrupted coordinates with good probability provided the fraction of corrupted coordinates~$\delta$ and query complexity~$q$ aren't too large.
The following definition allows us to make this intuition precise.

\begin{definition}[Smooth code]\label{def:smoothcode}
For positive integers~$k, n,q$ and parameters $\eta\in (0,1/2]$ and $c>0$, a map $C: \bset{k}\to\bset{n}$ is a \emph{$(q,c,\eta)$-smooth code} if, for every $i\in[k]$, there exists a randomized decoder $\mathcal A_i:\bset{n}\to\bset{}$ such that
\begin{itemize}
\item For every $x\in\bset{k}$,
	\beq\label{eq:smoothcode}
	\Pr\big[x_i=\cA_i\big(C(x)\big)\big]\ge \frac{1}{2}+\eta.
	\eeq
\item The decoder $\mathcal{A}_i$ (non-adaptively) queries at most $q$ coordinates of $C(x)$.
\item For each $j\in[n]$, the probability that $\mathcal A_i$ queries the coordinate~$j\in[n]$ is at most~$c/n$.
\end{itemize}
\end{definition}

The formal version of Katz and Trevisan's observation is as follows.

\begin{theorem}[Katz--Trevisan]
\label{thm:ldctosmooth}
If~$C:\bset{k}\to\bset{n}$ is a $(q,\delta,\eta)$-LDC, 
then~$C$ is also a $(q,q/\delta,\eta)$-smooth code. 
Conversely, if $C:\bset{k}\to\bset{n}$ is a $(q,c,\eta)$-smooth code, then~$C$ is also a $(q,\eta/2c,\eta/2)$-LDC. 
\end{theorem}

Our second step in the proof of Theorem~\ref{thm:main} is a stronger form of the converse part of Theorem~\ref{thm:ldctosmooth}. We show that even smooth codes that are only required to work \emph{on average} can be turned into LDCs, losing only a constant factor in the rate and success probability.

\begin{definition}[Average-case smooth code]
A code as in Definition~\ref{def:smoothcode} is a \emph{$(q,c,\eta)$-average-case} smooth code if instead of the first item,~\eqref{eq:smoothcode} is required to hold only on average over uniformly distributed $x\in\bset{k}$ and uniformly distributed~$i\in[k]$,
which is to say that
\beqn
\Pr\big[x_i = \mathcal A_i\big(C(x)\big)\big] \geq \frac{1}{2} + \eta,
\eeqn
where the probability is taken over $x$,~$i$ and the randomness used by~$\mathcal A_i$.
\end{definition}



\begin{lemma}\label{lem:ave2ldc}
Let $C:\bset{k}\to \bset{n}$ be a $(q,c,\eta)$-average-case smooth code. 
Then, there exists an $(q,\Omega(\eta/c), \Omega(\eta))$-LDC sending $\bset{l}$ to $\bset{n}$ where $l=\Omega(\eta^2 k/\log(1/\eta))$.
\end{lemma}

The idea behind the proof of Lemma~\ref{lem:ave2ldc} is as follows.
We first switch the message and codeword alphabets to~$\pmset{}$ and let $f_i:\pmset{k}\to [-1,1]$ be the expected decoding function $f_i(z)=\E[\cA_i(z)]$.
The properties of~$C$ then easily imply that the set $T\subseteq [-1,1]^k$ given by $T=\{(f_1(z),\dots,f_k(z)):z\in \sBC^k\}$ has large \emph{Gaussian width}, in particular it holds that for a standard $k$-dimensional Gaussian vector~$g$, we have $\Exp[\sup_{t\in T}\langle g, t\rangle]  \gtrsim \eps k$.\footnote{We write $A\gtrsim B$ and $A=\Omega(B)$ interchangeably to mean that $A\ge c B$ for some absoulte constant $c>0$ independent of all parameters involved.}
Next, we employ a powerful result of~\cite{Mendelson:2003} showing that~$T$ contains an $l$-dimensional hypercube-like structure with edge length some absolute constant~$c\in (0,1]$, for $l \gtrsim k$.
Roughly speaking, this implies that~$C$ is a smooth code on~$\pmset{l}$ whose decoding probability depends on~$\eps$ and~$c$.
Finally, we obtain an LDC via an application of Theorem~\ref{thm:ldctosmooth}.
The full proof is given in Section~\ref{sec:avg2ldc}.

\subsection{Organization}
Section~\ref{sec:prelims} contains some preliminaries in Fourier analysis over the Boolean cube. In Section~\ref{sec:outlaw2smoothcode}, we prove our main theorem (Theorem~\ref{thm:main}) by first showing that outlaw distributions over smooth functions imply existence of average-case smooth codes and using Lemma~\ref{lem:ave2ldc} to convert them to LDCs. In Section~\ref{sec:avg2ldc}, we prove Lemma~\ref{lem:ave2ldc} showing how to convert average-case smooth-codes to LDCs. In Section~\ref{sec:ldc2outlaw}, we show the converse to our main theorem (Theorem~\ref{thm:main_converse}) showing how to get outlaw distributions over smooth functions from LDCs. Finally in Section~\ref{sec:candidates}, we give some candidate constructions of outlaw distributions over smooth functions using incidence geometry and hypergraph pseudorandomness.

\section{Preliminaries}
\label{sec:prelims}

%
%
%

We recall a few basic definitions and facts from  analysis over the $n$-dimensional Boolean hypercube~$\pmset{n}$.
Equipped with the coordinate-wise multiplication operation, the hypercube forms an Abelian group whose group of characters is formed by the functions $\chi_S(x) = \prod_{i\in S}x_i$ for all  $S\subseteq[n]$.
The characters form a complete orthonormal basis for the space of real-valued functions on $\pmset{n}$ endowed with the inner product $\langle f,g\rangle = \Exp_{x\in\pmset{n}}[f(x)g(x)]$,
where we use the notation $\Exp_{a\in S}$ to denote the expectation with respect to a uniformly distributed element~$a$ over a set~$S$.
The \emph{Fourier transform} of a function $f:\pmset{n}\to\R$ is the function $\widehat f:2^{[n]}\to \R$ defined by $\widehat{f}(S) = \langle f,\chi_S\rangle$.
The Fourier inversion formula (which follows from orthonormality of the character functions) asserts that
\beqn
f = \sum_{S\subseteq [n]} \widehat{f}(S)\chi_S.
\eeqn
\emph{Parseval's Identity} relates the $L_2$-norms of~$f$ and its Fourier transform by
\beqn
\big(\Exp_{x\in\pmset{n}}[f(x)^2]\big)^{1/2}
=
\Big(\sum_{S\subseteq [n]}|\widehat{f}(S)|^2\Big)^{1/2}.
\eeqn
A function~$f$ has \emph{degree}~$q$ if~$\widehat{f}(S) = 0$ when $|S| > q$ and
the \emph{degree-$q$ truncation} of~$f$, denoted~$f^{\le q}$, is the degree-$q$ function defined by
\beqn
f^{\le q}=\sum_{|S|\le q}\widehat{f}(S)\chi_S.
\eeqn 
A function~$f$ is a \emph{$q$-junta} if it depends only on a subset of~$q$ of its variables, or equivalently, if there exists a subset $T\subseteq [n]$ of size~$|T|\leq q$ such that~$\widehat{f}(S) = 0$ for every $S\not\subseteq T$. The \emph{$i$th discrete derivative} $D_if$ is the function $(D_if)(x) = (f(x) - f(x^i))/2$, where $x^i$ is the point that differs from~$x$ on the $i$th coordinate.
It is easy to show that the $i$th discrete derivative in of a function~$f$ is given by
\beqn
D_if = \sum_{S\ni i} \widehat{f}(S) \chi_S.
\eeqn
Hence, it follows that $\|D_if\|_{\spec} = \sum_{S\ni i}|\widehat f(S)|$.


%
\section{From outlaws to LDCs}
\label{sec:outlaw2smoothcode}

In this section we prove Theorem~\ref{thm:main}.
For convenience, in the remainder of this paper, we switch the message and codeword alphabets of all codes from $\bset{n}$ to $\pmset{n}$.
We begin by showing that outlaw distributions over degree-$q$ functions give $q$-query average-case smooth codes.
Combined with Lemma~\ref{lem:ave2ldc}, this implies the second part of Theorem~\ref{thm:main}.

\begin{theorem}
\label{thm:degq_dist2avg}
Let $\mu$ be a probability distribution on $1$-smooth degree-$q$ functions on $\sBC^n$, let $\eps \in (0,1]$ and let $k = \kappa_\mu(\eps)$. 
Then, there exists a $(q,1,\eps/4)$-average-case smooth code sending $\sBC^k$ to $\sBC^n$.
\end{theorem}

\begin{proof}
The proof uses a symmetrization argument.
Let $\mathcal F = (f_1,\dots,f_k)$ and $\mathcal F' = (f_1',\dots,f_k')$ be two $k$-tuples of independent $\mu$-distributed random variables and let $\bar f = \Exp_\mu[f]$.
Then, by definition of~$\kappa_\mu(\eps)$ and Jensen's inequality,
\begin{align*}
\eps 
&\le 
\E_{\mathcal F}\Big[\Big\|\frac{1}{k}\sum_{i=1}^k (f_i - \bar f)\Big\|_{L_\infty}\Big]\\
&=
\E_{\mathcal F}\Big[\Big\|\frac{1}{k}\sum_{i=1}^k \big(f_i - \Exp_{\mathcal F'}[f_i']\big)\Big\|_{L_\infty}\Big]\\
&\le 
\E_{\mathcal{F,F'}}\Big[\Big\|\frac{1}{k}\sum_{i=1}^k (f_i - f_i')\Big\|_{L_{\infty}}\Big].
\end{align*}
The random variables $f_i - f_i'$ are symmetrically distributed, which is to say that they have the same distribution as their negations $f_i' - f_i$.
Since they are independent, it follows that
for every $x\in\pmset{k}$, the random variable $x_1(f_1 - f_1') + \cdots + x_k(f_k - f_k')$ has the same distribution as $(f_1 - f_1') + \cdots + (f_k - f_k')$.
Therefore,
\begin{align*}
\E_{\mathcal{F,F'}}\Big[\Big\|\frac{1}{k}\sum_{i=1}^k (f_i - f_i')\Big\|_{L_{\infty}}\Big]
&= 
\E_{x\in \sBC^k}\Big[\E_{\mathcal{F,F'}}\Big[\Big\|\frac{1}{k}\sum_{i=1}^k x_i(f_i - f_i')\Big\|_{L_{\infty}}\Big]\Big]\\
&\le
2\E_{\mathcal{F}}\Big[\E_{x\in \sBC^k}\Big[\Big\|\frac{1}{k}\sum_{i=1}^k x_if_i\Big\|_{L_{\infty}}\Big]\Big].
\end{align*}
Applying the Averaging Principle to the outer expectation, we find that there exist $1$-smooth degree-$q$ functions~$f_1^\star,\dots,f_k^\star:\pmset{n}\to\R$ such that 
\beq\label{eq:fstareps}
\E_{x\in \sBC^k}\Big[\Big\|\frac{1}{k}\sum_{i=1}^k x_if_i^\star\Big\|_{L_\infty}\Big]
\ge \frac{\eps}{2}.
\eeq

Define the code $C:\sBC^k\to \sBC^n$ such that for each $x\in\pmset{k}$, we have
\beq\label{eq:Cdef}
\frac{1}{k}\sum_{i=1}^k x_if_i^\star(C(x))=\Big\|\frac{1}{k}\sum_{i=1}^k x_if_i^\star\Big\|_{L_{\infty}}.
\eeq

For each~$i\in[k]$, define the decoder $\cA_i$ as follows.
Let $\nu_i:2^{[n]}\to[0,1]$ be the probability distribution defined by $\nu_i(S) = |\widehat{f_i^\star}(S)|/\|\widehat{f_i^\star}\|_{\Anorm}$.
Given a string~$z\in\pmset{n}$, with probability $1 -\|\widehat{f_i^\star}\|_{\Anorm}$, the decoder~$\mathcal A_i$ returns a uniformly random sign, and with probability~$\|\widehat{f_i^\star}\|_{\Anorm}$, it samples a set $S\subseteq[n]$ according to~$\nu_i$ and returns $\chi_S(z)$.
This is a valid probability distribution since for any $1$-smooth function~$f$, we have
\beqn
\specnorm{f}=\sum_{S\subset [n]}|\widehat{f}(S)|\le \sum_{S\subset [n]}|S||\widehat{f}(S)|=\sum_{i=1}^n\sum_{S\ni i} |\widehat{f}(S)| \le n\cdot \frac{1}{n}=1.
\eeqn
Then, $\mathcal A_i$ queries at most~$q$ coordinates of~$z$ and since $f_i^\star$ is $1$-smooth, the probability that it queries any coordinate $j\in[n]$ is at most $\|D_jf_i^\star\|_{\Anorm} \le 1/n$.
 We also have $\E[\cA_i(z)] = f_i^\star(z)$.
Therefore, by~\eqref{eq:fstareps} and~\eqref{eq:Cdef}, we have
\begin{align*}
\E_{x\in \sBC^k, i\in [k]}\left[\Pr[x_i=\cA_i(C(x))]\right]&=\frac{1}{2}+\frac{1}{2}\E_{x\in \sBC^k, i\in [k]}\left[x_i\E[\cA_i(C(x))]\right]\\
&=\frac{1}{2}+\frac{1}{2}\E_{x\in \sBC^k, i\in [k]}\left[x_if_i^\star(C(x))\right]\\
&=
\frac{1}{2}+\frac{1}{2}
\E_{x\in \sBC^k}\Big[\Big\|\frac{1}{k}\sum_{i=1}^k x_if_i^\star\Big\|_{L_\infty}\Big]\\
&\ge \frac{1}{2}+\frac{\eps}{4}.
\end{align*}
Hence, $C$ is a~$(q, 1, \eps/4)$-average-case smooth code.
\end{proof}


The final step before the proof of Theorem~\ref{thm:main} is to show that for any distribution~$\mu$ over smooth functions, there exists a distribution~$\tilde\mu$ over smooth functions of bounded degree that is not much more concentrated than~$\mu$.


\begin{lemma}
\label{lem:truncate_smooth_dist}
Let $\mu$ be a probability distribution over $1$-smooth functions and let  $\eps>0$.
Then, there exists a probability distribution $\tilde{\mu}$ over $1$-smooth functions of degree $q=4/\epsilon$ such that $\kappa_{\tilde{\mu}}(\epsilon/2)\ge \kappa_\mu(\epsilon)$.
\end{lemma}

\begin{proof}
We first establish that smooth functions have low-degree approximations in the supremum norm.
If $f:\sBC^n\to \R$ is $1$-smooth, then
\beqn
q\sum_{|S|>q}|\widehat{f}(S)|\le \sum_{S\subset [n]}|S||\widehat{f}(S)|=\sum_{i=1}^n\sum_{S\ni i} |\widehat{f}(S)| = \sum_{i=1}^n \specnorm{D_if}\le 1.
\eeqn
It follows that the degree-$q$ truncation~$f^{\leq q}$ satisfies
\beq\label{eq:truncdist}
\norm{f-f^{\le q}}_{L_\infty}\le \sum_{|S|>q} |\widehat{f}(S)|\le \frac{1}{q}=\frac{\epsilon}{4}.
\eeq
Define $\tilde{\mu}$ as follows: sample $f$ according to~$\mu$ and output~$f^{\le q}$. Clearly, $\tilde{\mu}$ is also a distribution over $1$-smooth functions. For $k=\kappa_\mu(\eps)$, we have 
\beqn
\Exp_{f_1,\dots,f_k \sim \mu}\left[\Big\|\frac{1}{k}\sum_{i=1}^k \big(f_i - \Exp[f_i]\big)\Big\|_{L_\infty}\right] \geq \eps.\eeqn
Hence, by the triangle inequality and~\eqref{eq:truncdist}, we have
\beqn
\Exp_{f_1,\dots,f_k \sim \tilde{\mu}}\left[\Big\|\frac{1}{k}\sum_{i=1}^k \big(f_i - \Exp[f_i]\big)\Big\|_{L_\infty}\right] \geq \frac{\epsilon}{2},
\eeqn
giving the claim.
\end{proof}

\begin{proof}[Proof of Theorem~\ref{thm:main}]
By applying Lemma~\ref{lem:truncate_smooth_dist} to $\mu$, we get a distribution $\tilde{\mu}$ over $1$-smooth degree $q=O(1/\epsilon)$ functions with $k'=\kappa_{\tilde{\mu}}(\epsilon/2)\ge \kappa_{\mu}(\epsilon)=k$. By Theorem~\ref{thm:degq_dist2avg}, we get a $(q,1,\Omega(\epsilon))$-average-case smooth code $C':\sBC^{k'}\to \sBC^n$. Finally we use Lemma~\ref{lem:ave2ldc} to convert~$C'$ to a $(q,\Omega(\epsilon),\Omega(\epsilon))$-LDC $C:\sBC^\ell \to \sBC^n$ where $\ell=\Omega(\epsilon^2k'/\log(1/\epsilon))$. 
For the last part of the theorem we can simply apply Theorem~\ref{thm:degq_dist2avg} directly.
\end{proof}

\section{From average-case smooth codes to LDCs}
\label{sec:avg2ldc}
In this section, we prove Lemma~\ref{lem:ave2ldc}. For this, we need the notion of the Vapnik--Chervonenkis dimension (VC-dimension).

\begin{definition}[VC-dimension]\label{def:vcdim}
For $T\subset [-1,1]^k$ and $w>0$, $\vc(T,w)$ is defined as the size of the largest subset $\sigma\subset [k]$ such that there exists a shift $s\in [-1,1]^k$ satisfying the following: for every $x\in \sBC^\sigma$, there exists $t\in T$ such that for every $i\in \sigma$, $(t_i-s_i)x_i\ge w/2$.
\end{definition}

Observe that if~$T$ is convex, then $\vc(T,w)$ is the maximum dimension of a shifted hypercube with edge lengths at least $w$ contained in~$T$ . 

\begin{definition}[Gaussian width]
Let~$g$ be a $k$-dimensional standard Gaussian vector, with independent standard normal distributed entries.
The \emph{Gaussian width} of a set~$T\subseteq \R^k$ is defined as
\beqn
E(T)=\E_g[\sup_{t\in T} \inpro{g}{t}].
\eeqn
\end{definition}

It is easy to see that a large VC-dimension implies a large Gaussian width.
The following theorem shows the converse: containing a hypercube-like structure is the only way to have large Gaussian width.

\begin{theorem}[\cite{Mendelson:2003}]\label{thm:vcdim}
Let $T\subset [-1,1]^k$. Then, the Gaussian width of $T$ is bounded as
$$E(T) \lesssim \sqrt{k} \int_{\alpha E(T)/k}^1 \sqrt{\vc(T,w) \log (1/w)} dw$$ for some absolute constant $\alpha >0$.
\end{theorem}

Finally, we use that fact that, as for LDCs, we can assume that on input~$y\in\bset{n}$, the decoder~$\mathcal A_i$ of a smooth code first samples a set~$S\subseteq[n]$ of at most~$q$ coordinates according to a probability distribution that depends on~$i$ only and then returns a random sign depending only on~$i$,~$S$ and the values of~$y$ at~$S$.

\begin{proof}[Proof of Lemma~\ref{lem:ave2ldc}]
The proof works by showing that the average-case smooth code property implies that the image of the (average) decoding functions should have large Gaussian width. We then use Theorem~\ref{thm:vcdim} to find a hypercube like structure inside the image, which we use to construct a smooth code. Finally we use Theorem~\ref{thm:ldctosmooth} to convert the smooth code to an LDC.

Recall the switch of the message and codeword alphabets to~$\pmset{}$.
For each $i\in[k]$, let $f_i:\pmset{n}\to [-1,1]$ be the expected decoding function $f_i(z)=\E[\cA_i(z)]$.
Let $g$ be a standard $k$-dimensional Gaussian vector and 
$T=\{(f_1(z),\dots,f_k(z)):z\in \sBC^n\}$.
By the definition of average-case smooth code we have 
\beqn
2\eta k\le \E_{x\in \sBC^k}\left[\sum_{i=1}^k x_if_i(C(x))\right]\le \E_{x\in \sBC^k}\left[\sup_{t\in T} \inpro{x}{t} \right] \lesssim \E_{g}\left[\sup_{t\in T} \inpro{g}{t} \right].
\eeqn 
(See for instance~\cite[Lemma~3.2.10]{Talagrand:2014} for the last inequality.)
By Theorem~\ref{thm:vcdim}, for some constant $\alpha>0$, we have 
\begin{align*}
\eta k \lesssim \sqrt{k} \int_{\alpha\eta}^1 \sqrt{\vc(T,w)\log(1/w)} dt \le \sqrt{k} \cdot \sqrt{\vc(T,\alpha\eta) \log(1/\alpha\eta)}
\end{align*} 
 where we used the fact that $\vc(T,w)$ is decreasing in $w$. So for $\tau=\alpha\eta$, we have $\vc(T,\tau)\gtrsim \eta^2 k/\log(1/\eta)$. By the definition of VC-dimension, there exists a subset $\sigma\subset [k]$ of size $|\sigma| \ge \vc(T,\tau)$ and a shift $s\in [-1,1]^k$ such that for every $x\in \sBC^\sigma$ there exists $t\in T$ such that $(t_i-s_i)x_i\ge \tau/2$ for every $i\in \sigma$.
 
Now we will define the code $C':\sBC^\sigma \to \sBC^n$. Given $x\in \sBC^\sigma$, there exists $t(x)\in T$ such that $(t(x)_i-s_i)x_i\ge \tau/2$ for every $i\in \sigma$. Define $C'(x)\in \pmset{n}$ to be one of the preimages of $t(x)$ under $f$, that is,
 $$\left(f_1(C'(x)),\dots,f_k(C'(x))\right)=t(x).$$ Let $W_p$ denote a $\sBC$-valued random variable with mean $p$. The decoding algorithms $\cA'_i(y)$ run $\cA_i(y)$ internally and give their output as follows: 
\[
\cA'_i(y)=
\begin{cases}
\text{Output } W_{(1-s_i)/2} & \text{if }\cA_i(y) \text{ returns } 1\\
\text{Output } -W_{(1+s_i)/2} & \text{if }\cA_i(y) \text{ returns } -1
\end{cases}
\]
 Therefore, for every $x\in \pmset{\sigma}$ and for every $i\in \sigma$,
\begin{align*}
x_i\E[\cA'_i(C'(x))] &=x_i\E\left[\frac{(1+\cA_i(C'(x)))}{2}W_{(1-s_i)/2}-\frac{(1-\cA_i(C'(x)))}{2}W_{(1+s_i)/2}\right]\\
&=\frac{x_i}{2}\E\left[\cA_i(C'(x))-s_i\right]\\
&=\frac{x_i}{2}(f_i(C'(x))-s_i)\\
&=\frac{x_i}{2}(t(x)_i-s_i)\\
&\ge \frac{\tau}{4}\gtrsim \eta.
\end{align*}
Since the probability that $\cA'_i(C'(x))$ queries any particular location of $C'(x)$ is still at most $c/n$, it follows that $C'$ is a $(q,c,\Omega(\eta))$-smooth code. By Theorem~\ref{thm:ldctosmooth}, $C'$ is also a $(q,\Omega(\eta/c),\Omega(\eta))$-LDC. 
  \end{proof}

\section{From LDCs to outlaws}
\label{sec:ldc2outlaw}

In this section we prove Theorem~\ref{thm:main_converse}, the converse of our main result.

\begin{proof}[Proof of Theorem~\ref{thm:main_converse}]
By Theorem~\ref{thm:ldctosmooth}, the map $C:\pmset{k}\to\pmset{n}$ is also a $(q,q/\delta,\eta)$-smooth code. For each~$i\in[k]$, let~$\mathcal B_i$ be its decoder for the $i$th index.
Let $\nu_i:2^{[n]}\to [0,1]$ be the probability distribution used by~$\mathcal B_i$ to sample a set~$S\subseteq[n]$ of at most~$q$ coordinates and 
 let~$f_{i,S}:\pmset{n}\to [-1,1]$ be function whose value at $y\in\pmset{n}$ is the expectation of the random sign returned by $\mathcal B_i(y)$ conditioned on the event that it samples~$S$. Since this value depends only on the coordinates in~$S$, the function~$f_{i,S}$ is a $q$-junta.

Fix an $i\in[k]$ and let~$f_i:\pmset{n}\to [-1,1]	$ be the function given by $f_i =\E_{S\sim\nu_i}[f_{i,S}]$.
Then, since a $q$-junta has degree at most~$q$, so does~$f_i$.
We claim that~$f_i$ is $\delta/(q2^{q/2})$-smooth.
Since the functions~$f_{i,S}:\pmset{n}\to\pmset{}$ are $q$-juntas, it follows from Parseval's identity that they have spectral norm at most $2^{q/2}$. 
Moreover, for each~$j\in[n]$, we have $\Pr_{S\sim\nu_i}[j\in S] \leq q/(\delta n)$.
Hence, since~$f_{i,S}$ depends only on the coordinates in~$S$, we have
\beqn
\specnorm{D_jf_i}
\leq
\sum_{S\ni j}\nu_i(S)\specnorm{f_{i,S}}
\leq
\frac{q2^{q/2}}{\delta n},
\eeqn
which gives the claim.
By~\eqref{eq:smoothcode}, it holds for every~$x\in\pmset{k}$ and every~$i\in[k]$ that
\beq\label{eq:smoothcode2}
x_if_i\big(C'(x)\big)\geq 2\eta.
\eeq

Define the distribution~$\mu$ to correspond to the process of sampling $i\in[k]$ uniformly at random and returning~$f_i$.
Let $\bar g = (f_1 + \cdots + f_k)/k$ be the mean of~$\mu$.
We show that $\kappa_\mu(\eta) \geq \eta k$.
To this end, let $l = \eta k$, let $\sigma:[l]\to[k]$ be an arbitrary map and define the functions~$g_1,\dots,g_l$ by $g_i = f_{\sigma(i)}$.
Let~$x\in\pmset{k}$ be such that for each~$i\in[l]$, we have $x_{\sigma(i)} = 1$ and $x_j = -1$ elsewhere.
It follows from~\eqref{eq:smoothcode2} that
$f_{\sigma(i)}\big(C(x)\big) \in [2\eta,1]$ for every $i\in[l]$ and that $f_{i}\big(C(x)\big) \leq 0$ for every other~$i\in[k]$.
Hence,
\begin{align*}
\Big\|
\frac{1}{l}\sum_{i=1}^l(g_i - \bar g)
\Big\|_{L_\infty}
&\geq
\Big(
\frac{1}{l}\sum_{i=1}^l(g_i - \bar g)
\Big)
\big(C(x)\big)\\
&=
\frac{1}{l}\sum_{i=1}^lf_{\sigma(i)}\big(C(x)\big)
-
\frac{1}{k}\sum_{i=1}^kf_{i}\big(C(x)\big)\\
&\geq
2\eta - \frac{l}{k}= \eta.
\end{align*}
If~$\sigma$ maps each element in~$[l]$ to a uniformly random element in~$[k]$, then $g_1,\dots,g_l$ are independent, $\mu$-distributed and satisfy
\begin{align*}
\Exp\left[
\Big\|
\frac{1}{l}\sum_{i=1}^l(g_i - \bar g)
\Big\|_{L_\infty}
\right]
\geq
\eta,
\end{align*}
which shows that $\kappa_\mu(\eta) \geq l$. Finally we can scale all the functions in $\mu$ to make them $1$-smooth, and get a distribution $\tilde{\mu}$ over $1$-smooth functions with $\kappa_{\tilde{\mu}}(\eta \delta/(q2^{q/2}))\ge \eta k$.
\end{proof}

\section{Candidate outlaws}
\label{sec:candidates}
In this section we elaborate on the candidate outlaws mentioned in the introduction.

\subsection{Incidence geometry}
\label{sec:lines}

We begin by describing a variant of Corollary~\ref{cor:lines_ldcs_informal} based on a slightly different assumption and show conditions under which this assumption holds.
Let~$p$ be an odd prime, let~$\F_p$ be a finite field with~$p$ elements and let~$n$ be a positive integer.
For $x,y\in \F_p^n$, the \emph{line} with origin~$x$ in direction~$y$, denoted~$\ell_{x,y}$, is the sequence $(x + \lambda y)_{\lambda \in \F_p}$.
A line is nontrivial if~$y\ne 0$.

\begin{corollary}\label{cor:lines_ldcs}
For every odd prime~$p$ and $\eps \in (0,1]$, there exist a positive integer $n_1(p,\eps)$ and a $c = c(p,\eps)\in (0,1/2]$ such that the following holds.
Let $n \geq n_1(p,\eps)$ and $k$ be positive integers.
Assume that for independent uniformly distributed elements $z_1,\dots,z_k\in \F_p^n$, with probability at least $1/2$, there exists a set $B\subseteq \F_p^n$ of size~$\eps p^n$ such that
every nontrivial line through the set~$\{z_1,\dots,z_k\}$ contains at most $p-2$ points of~$B$.
Then, there exists a $(p-1, c,c)$-LDC sending $\bset{l}$ to $\bset{p^n}$, where  $l=\Omega(c^2 k/\log(1/c))$.
\end{corollary}

The proof uses the following version of Szemer\'edi's Theorem~\cite[Theorem~1.5.4]{Tao:2012} and its standard ``Varnavides-type'' corollary (see for example~\cite[Exercise~10.1.9]{Tao:2006}).

\begin{theorem}[Szemer\'edi's theorem]\label{thm:szemeredi}
For every odd prime~$p$ and any~$\eps\in (0,1]$, there exists a positive integer $n_0(p,\eps)$ such that the following holds.
Let $n \geq n_0(p,\eps)$ and let~$S \subseteq \F_p^n$ be a set of size~$|S|  \geq \eps p^n$.
Then, $S$ contains a nontrivial line.
\end{theorem}

\begin{corollary}
\label{cor:varnavides}
For every odd prime~$p$ and any~$\eps\in (0,1]$, there exists a positive integer $n_1(p,\eps)$ and a $c(p,\eps)\in(0,1]$ such that the following holds.
Let $n \geq n_1(p,\eps)$ and let~$S \subseteq \F_p^n$ be a set of size~$|S|  \geq \eps p^n$.
Then, $S$ contains at least $c(p,\eps)p^{2n}$ nontrivial lines, that is,
$$\Pr_{x\in \F_p^n, y\in \F_p^n\smallsetminus\{0\}}\left[\{(x+\lambda y)_{\lambda=0}^{p-1}\}\subset S\right]\ge c(p,\epsilon).$$
\end{corollary}


\begin{proof}[Proof of Corollary~\ref{cor:lines_ldcs}]
Abusing notation, we identify functions $f:\F_p^n\to\pmset{}$ with vectors in~$\pmset{\F_p^n}$.
Let $\phi:\pmset{} \to\bset{}$ be the  map  $\phi(\alpha) = (\alpha+1)/2$.
For a function $f:\F_p^n\to\pmset{}$, let $\phi(f):\F_p^n\to\bset{}$ be the function $\phi(f)(x) = \phi(f(x))$ and for $f:\F_p^n\to\bset{}$, define $\phi^{-1}(f):\F_p^n\to\pmset{}$ analogously.

For every $x\in \F_p^n$, let~$F_{x}:\pmset{\F_p^n}\to\R$ be the degree-$(p-1)$ function
\beq\label{eq:Fxdef}
F_x(f) = \Exp_{y\in\F_p^n\smallsetminus\{0\}}\Bigg[\prod_{\lambda\in\F_p^*}\phi(f)(x + \lambda y)\Bigg].
\eeq
Then, for a set~$B\subseteq\F_p^n$, the value $F_x(\phi^{-1}(1_B))$ equals the fraction of all nontrivial lines~$\ell_{x,y}$ through~$x$ of which~$B$ contains the $p-1$ points $\{x + \lambda y\st \lambda\in \F_p^*\}$.
If~$B$ has size at least~$\eps p^n$, it follows from Corollary~\ref{cor:varnavides} that 
$\Exp_{x\in\F_p^n}[F_x(\phi^{-1}(1_B))] \geq c(p,\eps)$.
Moreover, since the monomials in the expectation of~\eqref{eq:Fxdef} involve disjoint sets of variables and can be expanded as
\beqn
\prod_{\lambda\in\F_p^*}\phi(f)(x + \lambda y)
=
\frac{1}{2^q}\sum_{S\subseteq \F_p^*}\prod_{\lambda\in S}f(x + \lambda y),
\eeqn
it follows that each~$F_x$ is $2(1 - p^{-n})$-smooth.

Let~$\mu$ be the uniform distribution over~$F_x$.
We claim that~$\kappa_\mu(c(p,\eps)) \geq k$, which implies the result by Theorem~\ref{thm:main} since~$\mu$ is supported by degree $(p-1)$-functions.
For every set~$A\subseteq \F_p^n$, let~$B_A\subseteq\F_p^n$ be a maximal set
such that
every nontrivial line through~$A$ contains at most $p-2$ points of~$B_A$,
and let~$f_A = \phi^{-1}(1_{B_A})$.
Let~$z$ be a uniformly distributed random variable over~$\F_p^n$, let  $z_1,\dots,z_k$ be independent copies of~$z$ and let~$A = \{z_1,\dots,z_k\}$.
Then, $F_{z_1},\dots,F_{z_k}$ are independent $\mu$-distributed random functions.
Moreover, in the event that both $|B_A| \geq \eps p^n$ and every nontrivial line through~$A$ meets~$B_A$ in at most~$p-2$ points, we have
\beqn
|(F_{z_i} - \Exp[F_z])(f_A)| 
=
\Exp\big[F_z(\phi^{-1}(1_{B_A}))\big] - F_{z_i}(\phi^{-1}(1_{B_A}))
 \geq c(p,\eps)
\eeqn
for every~$i\in[k]$.
Since this event happens with probability at least~$1/2$, we have
\begin{align*}
\Exp\Big[\Big\|\frac{1}{k}\sum_{i=1}^k\big(F_{z_i} - \Exp[F_{z}]\big)\Big\|_{L_\infty}\Big]
&\geq
\Exp\Big[\Big|\frac{1}{k}\Big(\sum_{i=1}^k\big(F_{z_i} - \Exp[F_{z}]\big)\Big)(f_A)\Big|\Big]
\geq
\frac{c(p,\eps)}{2},
\end{align*}
which gives the claim.
\end{proof}

The proof of the formal version of Corollary~\ref{cor:lines_ldcs_informal} (given below) is similar to that of Corollary~\ref{cor:lines_ldcs}, so we omit it.
In the following, $\PF_p^{n-1}$ is the projective space of dimension $n-1$, which is the space of directions in $\F_p^n$.
The formal version of Corollary~\ref{cor:lines_ldcs_informal} is then as follows.

\begin{corollary}\label{cor:forbidden_directions}
For every odd prime~$p$ and $\eps \in (0,1]$, there exist a positive integer $n_1(p,\eps)$ and a $c = c(p,\eps)\in (0,1/2]$ such that the following holds.
Let $n \geq n_1(p,\eps)$ and $k$ be positive integers.
Suppose that for independent uniformly distributed elements $z_1,\dots,z_k\in\PF_p^{n-1}$, with probability at least~$1/2$, there exists a set $B\subset \F_p^n$  of size $|B|\ge \eps p^n$ which does not contain any lines with direction in $\{z_1,\dots,z_k\}$.
Then, there exists a $(p, c,c)$-LDC sending $\bset{l}$ to $\bset{p^n}$, where $l=\Omega(c^2 k/\log(1/c))$.
\end{corollary}


%
\paragraph{Feasible parameters for Corollary~\ref{cor:lines_ldcs}}

Proving lower bounds on~$k$ for which the assumption of Corollary~\ref{cor:lines_ldcs} holds true thus allows one to infer the existence of $(p-1)$-query LDCs with rate~$\Omega(k/N)$ for $N = p^n$, provided~$p$ and~$\eps$ are constant with respect to~$n$.
We establish the following bounds, which imply the (well-known) existence of $(p-1)$-query LDCs with message length $k = \Omega((\log N)^{p-2})$.

\begin{theorem}\label{thm:RM2}
For every odd prime~$p$ there exists an~$\eps(p)\in (0,1]$ such that the following holds.
For every set $A\subseteq \F_p^n$ of size $|A| \leq {n+p-3\choose p-2}-1$, there exists a set $B\subseteq \F_p^n$ of size $\eps(p) p^n$ such that every line through~$A$ contains at most~$p-2$ points of~$B$.
\end{theorem}

The proof uses some basic properties of polynomials over finite fields.
For an $n$-variate polynomial $f\in \F_p[x_1,\dots,x_n]$ denote $Z(f) = \{x\in\F_p^n\st f(x) = 0\}$.
The starting point of the proof is the following standard result (see for example~\cite{Tao:2014}), showing that small sets can be `captured' by zero-sets of nonzero, homogeneous polynomials of low degree.

\begin{lemma}[Homogeneous Interpolation]\label{lem:hom-interpolation}
For every $A\subseteq \F_p^n$ of size $|A| \leq {n+d-1\choose d} - 1$,
there exists a nonzero homogeneous polynomial $f\in\F_p[x_1,\dots,x_n]$ of degree~$d$ such that $A\subseteq Z(f)$.
\end{lemma}

The next two lemmas show that if $f$ is nonzero, homogeneous and degree~$d$, and if $a\in\F_p^*$ is such that $f^{-1}(a)$ is nonempty, then  lines through~$Z(f)$ meet~$f^{-1}(a)$ in at most~$d$ points.

\begin{lemma}\label{lem:zerodirection}
Let $f\in\F_p[x_1,\dots,x_n]$ be a nonzero homogeneous polynomial of degree~$d$.
Let $a\in\F_p^*$ be such that the set $f^{-1}(a)$ is nonempty.
Then, every line that meets~$f^{-1}(a)$ in~$d+1$ points must have direction in~$Z(f)$.
\end{lemma}

\begin{proof}
The univariate polynomial~$g(\lambda) = f(x + \lambda y)$ formed by the restriction of~$f$ to  a line~$\ell_{x,y}$ has degree at most~$d$.
By the Factor Theorem, such a polynomial must be the constant polynomial $g(\lambda) = a$ to assume the value~$a$ for~$d+1$ values of~$\lambda$.
Since~$f$ is homogeneous, the coefficient of~$\lambda^d$, which must be zero, equals $f(y)$, giving the result.
\end{proof}

The following lemma is essentially contained in~\cite{BrietRao:2016}.

\begin{lemma}[Bri\"{e}t--Rao]\label{lem:zero-one}
Let $f\in\F_p[x_1,\dots,x_n]$ be a nonzero homogeneous polynomial of degree~$d$.
Let $a\in\F_p^*$ be such that $f^{-1}(a)$ is nonempty.
Then, there exists no line that intersects~$Z(f)$, meets~$f^{-1}(a)$ in at least~$d$ points and has direction in~$Z(f)$.
\end{lemma}

\begin{proof}
For a contradiction, suppose there exists a line $\ell_{x,y}$ through~$Z(f)$ that meets~$f^{-1}(a)$ in~$d$ points and has direction~$y\in Z(f)$.
Observe that for every $\lambda\in\F_p$, the shifted line $\ell_{x + \lambda y, y}$ also meets~$f^{-1}(a)$ in~$d$ points.
Hence, without loss of generality we may assume that the line starts in~$Z(f)$, that is~$x\in Z(f)$.
Let $g(\lambda) = a_0 + a_1\lambda + \cdots + a_d\lambda^d = f(x + \lambda y)\in\F_p[\lambda]$ be the restriction of~$f$ to~$\ell_{x,y}$.
It follows that $a_0 = g(0) = f(x) = 0$ and, since~$f$ is homogeneous, that~$a_d = f(y) = 0$.
Moreover, there exist distinct elements $\lambda_1,\dots,\lambda_d \in \F_p^*$ such that $g(\lambda_i) = f(x+\lambda_i y) = a$ for every $i\in[d]$. Then $g(\lambda)-a$ is a degree $d-1$ polynomial with $d$ distinct roots. But it cannot be the zero polynomial since it takes value $-a$ when $\lambda=0$.
\end{proof}

%

The final ingredient for the proof of Theorem~\ref{thm:RM2} is the DeMillo--Lipton--Schwartz--Zippel Lemma, as it appears in~\cite{CohenTal:2014}.

\begin{lemma}[DeMillo--Lipton--Schwartz--Zippel]\label{lem:DMLSZ}
Let~$f\in\F_p[x_1,\dots,x_n]$ be a nonzero polynomial of degree~$d$ and denote $r = |\F_p|$.
Then,
\beqn
|Z(f)| \leq \Big(1 - \frac{1}{r^{d/(r-1)}}\Big)r^n.
\eeqn
\end{lemma}

%
%
%

\begin{proof}[Proof of Theorem~\ref{thm:RM2}]
Let~$A\subseteq \F_p^n$ be a set of size~$|A|\leq {n+p-3\choose p-2}-1$.
Let~$f\in\F_p[x_1,\dots,x_n]$ be a nonzero degree-$(p-2)$ homogeneous polynomial such that~$A\subseteq Z(f)$, as promised to exist by Lemma~\ref{lem:hom-interpolation}.
By Lemma~\ref{lem:DMLSZ}, there exists an $a\in\F_p^*$ such that the set $B = f^{-1}(a)$ has size at least $|B| \geq p^n/p^{(2p - 3)/(p-1)}$.
By Lemma~\ref{lem:zerodirection}, every line that meets~$B$ in~$p-1$ points must have direction in~$Z(f)$, but by Lemma~\ref{lem:zero-one} no such line can pass through~$Z(f)$.
Hence, every line through~$A$ meets~$B$ in at most~$p-2$ points.
\end{proof}

\subsection{Hypergraph pseudorandomness}
\label{sec:ARLDCs}

A second candidate for constructing outlaws comes from special types of hypergraphs.
A hypergraph $H = (V,E)$ is a pair consisting of a finite vertex set~$V$ and an edge set~$E$ of subsets of~$V$ that allows for parallel (repeated) edges.
A hypergraph is $t$-uniform if all its edges have size~$t$.
For  subsets $W_1,\dots,W_t\subseteq V$,  define the induced edge count by
\begin{align*}
e_H(W_1,\dots,W_t)
&=
\sum_{v_1\in W_1}\cdots \sum_{v_t\in W_t} 
1_E(\{v_1,\dots,v_t\}).
\end{align*}
A perfect matching in a $t$-uniform hypergraph is a family of vertex-disjoint edges that intersects every vertex.
We shall use the following notion of pseudorandomness.

\begin{definition}[Relative pseudorandomness]
Let $H = (V,E)$, $J = (V, E')$ be $t$-uniform hypergraphs with identical vertex sets.
Then $J$ is $\eps$-pseudorandom relative to~$H$ if for all $W_1,\dots,W_t\subseteq V$, we have
\beq\label{eq:relqr}
\left|
\frac{e_J(W_1,\dots,W_t)}{|E'|} - \frac{e_H(W_1,\dots,W_t)}{|E|}
\right|
< \eps.
\eeq
\end{definition}

The left-hand side of~\eqref{eq:relqr} compares the fraction of edges that the sets $W_1,\dots,W_t$ induce in~$J$ with the fraction of edges they induce in~$H$.
Standard concentration arguments show that if $|E| \geq |V|$, then a random hypergraph $J$ whose edge set~$E'$ is formed by independently putting each edge of~$E$ in $E'$ with probability $p = p(\eps,t)$, is $\eps$-pseudorandom relative to~$H$ with high probability.
A deterministic hypergraph~$J$ is thus pseudorandom relative to~$H$ if it mimics this property of truly random sub-hypergraphs.
For graphs, relative $\eps$-pseudorandomness turns into a common notion sometimes referred to as $\eps$-uniformity when~$H$ is the complete graph with all loops, in which case~\eqref{eq:relqr} says that the number of edges induced by a pair of vertex-subsets $W_1,W_2$ is roughly equal to the product of their densities $(|W_1|/|V|)(|W_2|/|V|)$.
Uniformity in graphs is closely connected to the perhaps better-known notion of spectral expansion~\cite{Hoory:2006}.
These two notions were recently shown to be equivalent  (up-to universal constants) for all vertex-transitive graphs~\cite{Conlon:2016}.

We shall be interested in hypergraphs whose edge set can be partitioned into a family of ``blocks'', such that randomly removing relatively few of the blocks  likely leaves a hypergraph that is \emph{not} pseudorandom relative to the original.
(Think of a Jenga tower\footnote{\emph{Jenga}\textsuperscript{\textregistered} is a game of dexterity in which players begin with a tower of wooden blocks and take turns trying to remove a block without making the tower collapse.} that's already in a delicate balance, so that there are only few ways, or perhaps even no way, to remove many blocks without having it collapse.)
Our blocks will be formed by perfect matchings.
For technical reasons, the formal definition takes the view of building a new hypergraph out of randomly selected matchings, as opposed to obtaining one by randomly removing matchings.

\begin{definition}[Jenga hypergraph]
A $t$-uniform hypergraph~$H$ is $(k,\eps)$-jenga if its edge set can be partitioned into a family~$\mathcal M$ of perfect matchings
such that, with probability at least~$1/2$, the disjoint union of~$k$ independent uniformly distributed matchings from~$\mathcal M$ forms a hypergraph is \emph{not} $\eps$-pseudorandom relative to~$H$.
\end{definition}

We have the following simple corollary to Theorem~\ref{thm:main}.

\begin{corollary}\label{cor:jenga_ldc}
Let $n,k,t$ be positive integers and $\eps \in (0,1]$.
Assume that there exists a $t$-uniform $n$-vertex hypergraph that is $(k, \eps)$-jenga.
Then, there exists a $(t, 1, \Omega(\eps/t^2))$-LDC sending $\bset{l}$ to $\bset{tn}$, where $l = \Omega(\eps^2k/t^4\log(t^2/\eps))$.
\end{corollary}

\begin{proof}
Let $H = (V,E)$ be a hypergraph as assumed in the corollary.
Let~$\mathcal M$ be a partition of~$E$ into perfect matchings such that if $M_1,\dots,M_k$ are independent and uniformly distributed over~$\mathcal M$, then with probability at least $1/2$, the hypergraph $J = (V, M_1\uplus\cdots\uplus M_k)$ is not $\eps$-pseudorandom relative to~$H$.

Let~$V_1,\dots,V_t$ be copies of~$V$.
For each $M\in \mathcal M$, define $f_M:\R^{V_1\cup\cdots\cup V_t}\to\R$ by
\beqn
f_M(x[1],\dots,x[t]) = \frac{1}{|M|}
\sum_{v_1\in V_1}\cdots\sum_{v_t\in V_t}
1_M(\{v_1,\dots,v_t\})\,
x[1]_{v_1}\cdots x[t]_{v_t},
\quad
x[i] \in \R^{V_i}.
\eeqn
The function $f_M$ is a degree-$t$ polynomial.
Since every one of the $tn$ variables appears in exactly one monomial and $|M| = n/t$, the restriction of~$f_M$ to $\pmset{V_1\cup\cdots\cup V_t}$ is $t^2$-smooth.
Moreover, for $J = (V, M)$ and $W_1,\dots,W_t\subseteq V$, we have
\beqn
f_{M}(1_{W_1},\dots, 1_{W_t})
=
\frac{e_J(W_1,\dots,W_t)}{|M|}.
\eeqn

Let $M_1,\dots,M_k$ be independent uniformly distributed matchings from~$\mathcal M$ and consider the random  hypergraph $J = (V, M_1\uplus\cdots\uplus M_k)$.
Let $\bar f = \Exp[f_{M_1}]$ be the expectation of the random function~$f_{M_1}$ and note that $\Exp[f_{M_i}]  = \bar f$ for each~$i\in[k]$.
Then, since the functions $f_{M_i} - \bar f$ are multilinear,
\begin{align*}
\Exp\Big[
\Big\|
\frac{1}{k}\sum_{i=1}^k (f_{M_i} - \bar f)
\Big\|_{L_\infty}
\Big] 
&\geq
\Exp\Big[
\max_{W_1,\dots,W_t\subseteq V}
\Big|
\frac{1}{k}\sum_{i=1}^k (f_{M_i} - \bar f)
(1_{W_1},\dots,1_{W_t})
\Big|
\Big]\\
&=
\Exp\Big[
\max_{W_1,\dots,W_t\subseteq V}
\Big|
\frac{e_J(W_1,\dots,W_t)}{k|M|}
-
\frac{e_H(W_1,\dots,W_t)}{|E|}\Big|
\Big]\\
&\geq
\frac{\eps}{2}.
\end{align*}

The result now follows from Theorem~\ref{thm:main}.
\end{proof}

In the context of outlaws and LDCs, the relevant question concerning Jenga hypergraphs is the following.
Let $\kappa^J(n,t,\eps)$ denote the maximum integer~$k$ such that there exists an $n$-vertex $t$-uniform hypergraph that is $(k,\eps)$-jenga.

\begin{question}\label{prob:jenga}
For integer $t \geq 2$ and parameter $\eps\in (0,1]$, what is the growth rate of $\kappa^J(n,t,\eps)$ as a function of~$n\in\N$?
\end{question}

For $t = 2$ (graphs), the answer to Question~\ref{prob:jenga}  follows from famous work of Alon and Roichman~\cite{Alon:1994a} on expansion of random Cayley graphs, which implies that for constant $\eps\in (0,1]$, we have $\kappa(n, 2,\eps) = \Theta(\log n)$.
The lower bound follows for instance by partitioning the edge set of the complete graph with vertex set $V = \F_2^m$ into the collection of matchings of the form $M_y = \big\{\{x, x+y\} \st x\in \F_2^m\big\}$ for each $y\in \F_2^m\smallsetminus\{0\}$.
Any $m-1$ of such matchings give a graph with two disconnected components of equal size, making it $(m-1, \tfrac14)$-jenga.
Via Corollary~\ref{cor:jenga_ldc}, this arguably gives the most round-about way to prove the existence of 2-query LDCs matching the paramaters of the Hadamard code!
Generalizing the above example, \cite{BrietRao:2016} considered the $p$-uniform hypergraph on $\F_p^m$ whose edges are the (unordered) nontrivial lines.
It was shown that this hypergraph is $(m^{p-1}, \eps)$-jenga for some $\eps =  \eps(p)$ depending on~$p$ only, by partitioning the edge set according to the directions of the lines, that is, partitioning it with the matchings $M_y = \big\{\{x + \lambda y\st \lambda\in \F_p\} \st x\in \F_p^m\big\}$, $y\in \F_p^m\smallsetminus\{0\}$.
To the best of our knowledge, the best upper bounds on $\kappa^J(n,t,\eps)$ for constant $t\geq 3$ and $\eps \in (0,1]$ follow from upper bounds on LDCs, via Corollary~\ref{cor:jenga_ldc}.

We end with the following natural question concerning Jenga hypergraphs.

\begin{question}
Is $\kappa^J(n,t,\eps)$ largest for the complete hypergraph?
\end{question}

\paragraph{Acknowledgements}
J.B. and S.G. thank the Simons Institute for hosting them during the 2017 Pseudorandomness program, where part of this work was done.

\bibliographystyle{alpha}
\bibliography{randldc}

\end{document}